\documentclass[11pt]{article}
\usepackage[utf8]{inputenc}

\usepackage{amsmath,amsthm,nicefrac}

\usepackage{setspace}
\usepackage[breaklinks]{hyperref}
\hypersetup{colorlinks=true,
            citebordercolor={.6 .6 .6},linkbordercolor={.6 .6 .6},%
citecolor=blue,urlcolor=black,linkcolor=blue,pagecolor=black}
\usepackage[compact]{titlesec}
\usepackage[nameinlink]{cleveref}
\Crefname{algocf}{Algorithm}{Algorithms}
\crefname{algocfline}{line}{lines}
\Crefname{invariant}{Invariant}{Invariants}

\usepackage{epsfig}
\usepackage{amsthm,amssymb,mathrsfs}
\usepackage{xspace}
\usepackage{soul}
\usepackage{latexsym}
\usepackage{bbm}
\usepackage{enumitem}
\usepackage{thm-restate}

\usepackage[dvipsnames]{xcolor}
\definecolor{DarkGray}{rgb}{0.66, 0.66, 0.66}
\definecolor{DarkPowderBlue}{rgb}{0.0, 0.2, 0.6}
\definecolor{fluorescentyellow}{rgb}{0.8, 1.0, 0.0}

\usepackage[ruled,vlined,linesnumbered,algonl]{algorithm2e}
\SetEndCharOfAlgoLine{}
\SetKwComment{Comment}{\footnotesize$\triangleright$\ }{}

\SetCommentSty{mycommfont}

\usepackage{thmtools,thm-restate}

\usepackage[letterpaper,top=2cm,bottom=2cm,left=3cm,right=3cm,marginparwidth=1.75cm]{geometry}
\makeatletter
\allowdisplaybreaks

\newcounter{note}[section]
\renewcommand{\thenote}{\thesection.\arabic{note}}
\ifdefined\DEBUG
\newcommand{\alert}[1]{{\color{red}#1}}
\newcommand{\agnote}[1]{\refstepcounter{note}$\ll${\bf Anupam~\thenote:}
  {\sf \color{blue} #1}$\gg$\marginpar{\tiny\bf AG~\thenote}}
\newcommand{\elnote}[1]{\refstepcounter{note}$\ll${\bf E~\thenote:}
  {\sf \color{gray} #1}$\gg$\marginpar{\tiny\bf EL~\thenote}}
\newcommand{\dpnote}[1]{\refstepcounter{note}$\ll${\bf Debmalya~\thenote:}
  {\sf \color{magenta} #1}$\gg$\marginpar{\tiny\bf DP~\thenote}}  
\else
\newcommand{\alert}[1]{}
\newcommand{\agnote}[1]{}
\newcommand{\elnote}[1]{}
\newcommand{\dpnote}[1]{}
 \fi

\newcommand{\initOneLiners}{%
    \setlength{\itemsep}{0pt}
    \setlength{\parsep }{0pt}
    \setlength{\topsep }{0pt}
}

\newtheorem{theorem}{Theorem}[section]
\newtheorem{lemma}[theorem]{Lemma}

\newtheorem{defn}[theorem]{Definition}

\newcommand{\ip}[1]{\langle #1 \rangle}
\newcommand{\eat}[1]{}

\newcommand{\eps}{\varepsilon}

\newcommand{\EE}{\mathbb{E}}

\newcommand{\R}{\mathbb{R}}
\newcommand{\E}{\mathbb{E}}
\newcommand{\N}{\mathbb{N}}

\newcommand{\cS}{\mathcal{S}}

\newcommand{\sign}{\operatorname{sign}}

\newcommand{\poly}{\operatorname{poly}}

\newcommand{\nf}{\nicefrac}

\newcommand{\MC}{\textsc{MaxCut}\xspace}
\newcommand{\SC}{\textsc{SparsestCut}\xspace}

\newcommand{\alphaGW}{\alpha_{\text{GW}}}
\newcommand{\alphaRT}{\alpha_{\text{RT}}}

\newcommand{\opt}{{\rm opt}\xspace}
\newcommand{\sdp}{{\rm sdp}\xspace}

\newcommand{\cI}{\mathcal I}

\newcommand{\Abs}[1]{\left\lvert#1\right\rvert}
\newcommand{\set}[1]{\{#1\}}

\newcommand{\val}{\textnormal{val}}

\newcommand{\Paren}[1]{\left(#1\right)}
\newcommand{\brac}[1]{[#1]}
\newcommand{\Brac}[1]{\left[#1\right]}

\DeclareMathOperator*{\Var}{\textnormal{var}}

\newcommand{\bbP}{\mathbb P}
\newcommand{\cE}{\mathcal E}

\title{Max-Cut with $\eps$-Accurate Predictions}

\author{ Vincent Cohen-Addad\thanks{Google Research France.}  \and Tommaso  d'Orsi\thanks{Bocconi University. Part of this work was done while the author was at ETH Z\"urich.} \and Anupam Gupta\thanks{New York University. Supported in part by NSF awards CCF-2006953 and CCF-2224718, and by Google, Inc. Part of this work was done while the author was at Carnegie Mellon University. }  \and Euiwoong Lee\thanks{University of Michigan. Supported in part by NSF award CCF-2236669 and by Google, Inc.}  \and Debmalya Panigrahi\thanks{Duke University.}} 
\date{}

\begin{document}


\maketitle

\begin{abstract}
  We study the approximability of the \textsc{MaxCut} problem in the
  presence of predictions. 
  Specifically, we consider two models:
  in the \emph{noisy predictions} model, for each vertex we are given
  its correct label in $\{-1,+1\}$ with some unknown
  probability $\nf12 + \eps$, and the other (incorrect) label
  otherwise. In the more-informative \emph{partial predictions} model, for each vertex we are given its correct label with probability $\eps$ and no label otherwise. We assume only pairwise independence between vertices in both models. 

  We show how these predictions can be used to improve on the
  worst-case approximation ratios for this problem. Specifically, we
  give an algorithm that achieves an $\alpha + \widetilde{\Omega}(\eps^4)$-approximation
  for the noisy predictions model, where $\alpha \approx 0.878$ is the
  \textsc{MaxCut} threshold. While this result also holds for the
  partial predictions model, we can also give a $\beta +
  \Omega(\eps)$-approximation, where $\beta \approx 0.858$ is the approximation ratio
  for \textsc{MaxBisection} given by Raghavendra and Tan. This answers a question posed
  by Ola Svensson in his plenary session talk at SODA'23.
\end{abstract}

\thispagestyle{empty}

\newpage
\setcounter{page}{1}

\maketitle

\section{Introduction}
\label{sec:introduction}

The study of graph cuts has played a central role in the evolution and success of the field of algorithm design since its early days. In particular, cut problems have served as a testbed for models and techniques in ``beyond worst case'' algorithm design~\cite{boppana1987eigenvalues, R2020}, as researchers have striven to bridge the gap between observed real-world performance of algorithms and their theoretical analyses based on worst-case instances. Over the years, this has given rise to several deep lines of research, such as the study of (semi-)random instances of cut problems (e.g.,~\cite{pmlr-v35-mossel14,newman2006modularity,DBLP:conf/soda/CarsonI01,abbe2015exact,abbe2017community,cohen2020power}) and the exploration of optimal cuts that are stable to random noise (introduced by Bilu and Linial~\cite{bilu2012stable}, see also~\cite{makarychev2014bilu}). In recent years, the abundance of data and the impact of machine learning has led to algorithmic models that seek to go beyond worst-case performance using a {\em noisy prediction} of an optimal solution, typically generated by a machine learning model or a human expert or even by crowdsourcing. This paradigm has been particularly successful at overcoming information-theoretic barriers in online algorithms (see the CACM article by Mitzenmacher and Vassilvitskii~\cite{MitzenmacherV22}) 
but is also a natural alternative for transcending computational barriers. 
Motivated by this vision, in his SODA '23 plenary lecture, Ola Svensson posed the following question: {\em In the \MC problem, suppose we are given a prediction for the optimal cut that is independently correct for every vertex with probability $\nf 12 +\eps$. 
Can we exploit this information to breach the \MC threshold of $\alphaGW \simeq 0.878$ and obtain an $(\alphaGW+f(\eps))$-approximate solution?}

In this paper, we give an affirmative answer to this question {\em for all} $\eps > 0$. Namely, we give an algorithm that for any $\eps > 0$, obtains an $(\alphaGW+\tilde{\Omega}(\eps^4))$-approximate \MC solution. Furthermore, we relax the independence requirement to just pairwise independence of the predictions on the vertices. We further complement this result by considering another natural prediction model where instead of a noisy prediction for every vertex, we get a correct prediction but only for an $\eps$-fraction of randomly chosen vertices. In this case, we obtain an $(\alphaRT+\Omega(\eps))$-approximate solution to \MC, where $\alphaRT\simeq 0.858$ is the approximation factor obtained by Raghavendra and Tan for the \textsc{MaxBisection} problem~\cite{RT12}. Note that $\alphaRT$ is slightly smaller than $\alphaGW$, but we get a better advantage of $\Omega(\eps)$ instead of $\Omega(\eps^4)$. 

{\bf The \MC Problem.} We start by describing the \MC problem. In this problem, we are given a weighted graph
$G = (V,E)$ represented by a (symmetric) $n\times n$ adjacency matrix
$A$, where $A_{ij} = w_{ij}$, the weight of edge $\{i, j\}$ if it exists,
and $0$ otherwise. (We assume the graph has no
self-loops, and hence $A$ has zeroes on the diagonal.) We use $D$ to
denote the diagonal matrix $D_{ii} = \sum_{j\in [n]} A_{ij}$ and
$L = D-A$ to denote the (unnormalized) Laplacian matrix of the graph.
Note that $x \in \{-1,1\}^n$ denotes a cut in the graph, and the quadratic form
\[ \ip{x, Lx} = \sum_{\{i,j\} \in E} w_{ij} (x_i - x_j)^2 \] counts (four
times) the weight of edges crossing the cut between the vertices
labeled $1$, and those labeled $-1$. Hence, \MC can be rephrased as
follows:
\[  \text{\MC}(G) := \max_{x\in \{-1, 1\}^n} \nf14 \cdot \ip{x, Lx}. \]

\subsection{The Noisy/Partial Predictions Framework}

In this work, we initiate a study of cut problems under two different
models with predictions that may be incorrect or incomplete. In defining these models, by the ``correct solution'', we mean an (arbitrary) fixed  optimal solution for the \MC instance. 

\begin{enumerate}
\item The \emph{noisy predictions model}, where for each vertex we are
  given its correct label in $\{-1, +1\}$ with some
  unknown probability $\nf12 + \varepsilon$, and the other (incorrect)
  label otherwise. Here we only assume pairwise independence; for any two vertices $i$ and $j$, $\Pr[i, j\text{ both give their correct labels}] = (\nf12 + \eps)^2$.
  
\item The \emph{partial predictions model}, where for each vertex we are given its correct label in $\{-1, +1\}$ with some
  unknown probability $\varepsilon$, and no (blank) label otherwise. Again, we only assume pairwise independence;
  for any two vertices $i$ and $j$, $\Pr[i, j\text{ both give their}\\ \text{correct labels}] = \eps^2$.
\end{enumerate}

These two models are very basic and aim at capturing scenarios where
we are given very noisy predictions. 
These models can
easily be extended to other cut problems (\SC, \textsc{MinBisection},
\textsc{CorrelationClustering}, etc.), and more generally to other CSPs. In this paper, we
explore the power of such predictions to approximate \MC instances:
clearly we can just ignore the predictions and use 
Goemans-Williamson rounding~\cite{GoemansW95} to obtain an $\alphaGW \approx 0.878$-approximation for
any instance---but \emph{can we do better}?  

\subsection{Our Results and Techniques}

We show how to answer the above question positively in both prediction
models. In the more demanding model, we get noisy predictions. Here we
can just output the prediction to beat the random cut bound of $m/2$
by $O(\eps^2(\opt - m/2))$, but this is generally worse than
the SDP-based cut algorithm. The following theorem
shows how to outperform $\alphaGW$ by an
additive $\poly(\eps)$ factor:

\begin{restatable}[Noisy Predictions]{theorem}{NoisyThm}
  \label{thm:gen-degree}
  Given noisy predictions with a bias of $\eps$,
  there is a polynomial-time randomized algorithm that obtains 
  an approximation factor of $\alphaGW + \tilde{\Omega}(\eps^4)$
  in expectation for the \MC problem. 
\end{restatable}


The intuition for this theorem comes from
considering graphs with high degree: when $x^* \in \{ -1, +1 \}^n$
denotes the optimal solution and a vertex $i$ has a degree at least
$\Omega(n)$, the celebrated PTAS of Arora, Karger, and
Karpinski~\cite{aroraKK99} starts by obtaining a good estimate of
$(Ax^*)_i$ by sampling $\widetilde{O}(1)$ uniformly random vertices
and exhaustively guessing their $x^*$ values. We observe that under
$\eps$-noisy predictions where $\eps n$ {\em bits of
  information} 
is given, the degree of a vertex $i$ only needs to be least $\Omega(1/\eps^2)$ to obtain a similarly good estimate on $(Ax^*)_i$ and eventually a PTAS. 
For low-degree graphs, we use the algorithm by Feige, Karpinski, and Langberg~\cite{FeigeKL02}, whose analysis was recently refined by Hsieh and Kothari~\cite{HsiehK22} to guarantee an $(\alphaGW+\widetilde{O}(\nf{1}{d^2}))$-approximation where $d$ is the maximum degree. We give details of \Cref{thm:gen-degree} in \S\ref{sec:AKK}.

The AKK result applies not just to \MC but to ``dense'' instances of all 2-CSPs. We show that this generalization extends to the noisy predictions model. In particular, for dense instances of 2-CSPs, a noisy prediction with bias $\eps$ can be leveraged to obtain a PTAS where the density threshold is a function of the prediction bias and the approximation error. This extension to general 2-CSPs is given in \S\ref{section:csps}.

Our next result is for the more-informative partial predictions model.
In this model, we show how to easily get an additive improvement of
$O(\eps^2)$ over $\alphaGW$, and again ask: can we do better? Our
next result gives an affirmative answer:

\begin{restatable}[Partial Predictions]{theorem}{PartialThm}
  \label{thm:partial-info-rt}
  Given partial predictions with a rate of $\eps$, there is a
  polynomial-time randomized algorithm that obtains an (expected) approximation
  factor of
  $\alphaRT + (1 - \alphaRT - o(1))(2\eps - \eps^2)$ for the \MC problem.
\end{restatable}

Here $\alphaRT$ is the approximation factor for \textsc{MaxBisection}
given by Raghavendra and Tan~\cite{RT12}, and approximately $0.858$, a quantity
slightly smaller than $\alphaGW$. The advantage of their rounding
algorithm is that it preserves marginals of vertices; the probability that a vertex $i$ becomes $1$ is exactly the quantity predicted by the SDP solution.

\subsection{Related Work}


Since the early days, bringing together theory and practice has been
the focus of an intense research effort. The driving observation is
that the best theoretical algorithms are often designed to handle
worst-case instances that rarely occur in practice (and the successful
practical heuristics do not have provable guarantee in the
\emph{worst-case}). This observation has thus given birth to a line of
work on the so-called \emph{beyond worst-case} complexity of various
optimization problems~\cite{boppana1987eigenvalues}, and in particular cut
problems~\cite{R2020}. One iconic example is the \emph{Stochastic
  Block Model}, which provides a distribution over graphs that
exhibits a ``clear'' ground-truth cut, and the goal is to design
algorithms that on input drawn from the stochastic block model,
identify the ground-truth cut. This has given rise to a deep line of
work on the complexity of random instances for cut problems and has
allowed to show that some heuristics used in practice perform well on
such inputs~\cite{pmlr-v35-mossel14,newman2006modularity,DBLP:conf/soda/CarsonI01,abbe2015exact,abbe2017community,cohen2020power}.  This work has been generalized to random
Constraint Satisfaction Problems (CSP) instances, hence contributing
identifying and characterizing hard CSP instances (e.g.:~\cite{DBLP:conf/stoc/GuruswamiKM22}). Another
celebrated line of work is provided by the stability notion introduced
by Bilu and Linial~\cite{bilu2012stable} who consider \MC instances (and other
combinatorial optimization problems) where the optimum solution is
unique and remains optimal for any graph obtained by a
$\gamma$-perturbation of the edge weights. Makarychev, Makarychev, and
Vijayaraghavan~\cite{makarychev2014bilu} showed that it is possible to find an exact
optimal solution if $\gamma = \Omega( \sqrt{\log n}\log \log n)$ and
hard if $\gamma = o(\alpha_{SC})$, where $\alpha_{SC}$ denotes the
best polytime approximation ratio for the \SC problem.

More recently, the abundance of data and the impact of machine
learning has led us to design other models to capture new emerging
real-world scenarios and bridge the gap between theory and practice.
For clustering inputs, Ashtiani, Kushagra and
Ben-David~\cite{ashtiani2016clustering} introduced a model that
assumes that the algorithm can query an external oracle (representing,
for example, a machine learning model) that provides the optimum
cluster labels; the goal is then to limit the number of queries to the
oracle. This has led to several results leading to tight bounds for
various clustering objectives, from
$k$-means~\cite{DBLP:conf/innovations/AilonBJ018} to correlation
clustering~\cite{mazumdar2017clustering}. 
In some of these and
in follow-up works, e.g., Green Larsen et
al.~\cite{DBLP:conf/www/LarsenMT20}, and Del Pia et
al.~\cite{del2022clustering}, researchers have considered a more robust setting where the same-cluster oracle answers are noisy. Another set
of works is by Ergun et al.~\cite{ErgunFSWZ22}, and Gamlath et
al.~\cite{GamlathLNS22}, who consider $k$-means and related clustering
problems where the node labels are noisy. E.g., \cite{GamlathLNS22}
showed that even when the cluster labels provided by the oracles are
correct with a tiny probability (say $1\%$), it is possible to
obtain a $(1+o(1))$-approximation to the $k$-means objective as long
as the clusters are not too small. Noisy predictions similar to ours for online problems such as caching and online covering have also been considered (e.g.,~\cite{GPSS-neurips}). More generally, various online problems have been considered in the prediction model, and the dependence of algorithmic performance on a variety of noise parameters has been explored (see, e.g., Lykouris and Vassilvitskii~\cite{LykourisV21} who formalized this setting and the many papers listed at \cite{ALPS}).

A variant of the model considered in this paper has also been considered in an independent work by Bampis, Escoffier, and Xefteris~\cite{bampis2024parsimonious} where they focus
  on several cut problems (among them \MC) in dense graphs
  and provide an algorithm that approximates \MC within a 
  ratio $1-\eps - O(\text{e})$ factor if the number of edges  is $\Omega(n^2)$, where e is the 
  probability of a vertex to be mislabeled in a sample of
  size  of $O(\log n/\eps^3)$.




\subsection{Notation and Preliminaries}
\label{sec:prelims}

As mentioned in the introduction, the weighted undirected graph
$G = (V,E)$ is represented by an $n\times n$ weighted adjacency matrix
$A$, where $A_{ij} = w_{ij} \mathbbm{1}_{\{i,j\} \in E}$. The matrix
$A$ has zero diagonals. We let $W_i = \sum_j w_{ij}$ denote the
weighted degree of vertex $i$, let
$D := \text{diag}(W_1, \ldots, W_n)$ be the diagonal matrix with these
weighted degrees, and hence $L := D - A$ be the weighted
(unnormalized) Laplacian matrix of the graph.

In the \emph{noisy predictions} model, we assume there is some fixed
and unknown optimal solution $x^*\in \{-1, 1\}^n$. The algorithm has
access to a {\em prediction vector} $Y \in \{-1,1\}^n$, such that
prediction $Y_i$ is pairwise-independently {\em correct} with
some (unknown) probability $p$; namely, for each $i$ we have
\begin{gather}
 \Pr[Y_i = x^*_i] = p \qquad\text{and}\qquad\Pr[Y_i = -x^*_i] = 1-p \label{eq:pred-noisy},
\end{gather}
and for each $i \neq j$, 
\[
\Pr[Y_i = x^*_i \text{ and } Y_j = x^*_j] = p^2.
\]
We assume that the predictions have a strictly positive bias, i.e.,
$p > \nf 12$, and we denote the {\em bias} of the predictions by
$\eps := p - \nf 12$. For now we assume that the biases of all nodes
are exactly $\eps$: while this can be relaxed a bit, 
allowing arbitrarily different biases for different nodes may require
new ideas. (See the discussion at the end of the paper.) 
Also, we can assume we know $\eps$, since we can run our algorithm for
multiple values of $\eps$ chosen from a fine grid over the interval
$(0,\nf12)$ and choose the best cut among these runs. 

In the \emph{partial predictions} model with \emph{rate} $\eps$ we are given
a \emph{prediction vector} $Y \in \{-1, 0, 1\}^n$ where for each vertex $i \in V$, we have
\begin{gather}
  \Pr[Y_i = x^*_i] = \eps \qquad\text{and}\qquad\Pr[Y_i = 0] =
  1-\eps.\label{eq:pred-partial}
\end{gather} and  for each $i \neq j$, 
\[
\Pr[Y_i = x^*_i \text{ and } Y_j = x^*_j] = \eps^2.
\]


\newcommand{\hr}{\hat{r}}
\newcommand{\hx}{\hat{x}}
\newcommand{\hv}{\hat{v}}
\newcommand{\Vn}{V_{< \Delta}}
\newcommand{\Vw}{V_{> \Delta}}
\newcommand{\Wn}{W_{< \Delta}}
\newcommand{\Ww}{W_{> \Delta}}
\newcommand{\tA}{\tilde{A}}
\newcommand{\hA}{\hat{A}}

\section{\MC in the Noisy Prediction Model}
\label{sec:AKK}

Our goal in this section is to show the following theorem:

\NoisyThm*


A basic distinction that we will use throughout this section is that
of $\Delta$-wide and $\Delta$-narrow graphs; these should be thought
of as weighted analogs of high-degree and low-degree graphs. We first define these and
related concepts below, then we present an algorithm for the \MC problem on
$\Delta$-wide graphs in \S\ref{sec:high-degree}, followed by the
result for $\Delta$-narrow graphs in 
\S\ref{sec:low-degree}. We finally wrap up with the proof of
\Cref{thm:gen-degree}.

We partition the edges
incident to vertex $i$ into two sets: the {\em $\Delta$-prefix} for
$i$ comprises the $\Delta$ heaviest edges incident to $i$ (breaking
ties arbitrarily), while the remaining edges make up the {\em
  $\Delta$-suffix} for $i$. We fix a parameter $\eta \in
(0,\nf12)$. We will eventually set $\Delta = \Theta(1/\eps^2)$ and
$\eta$ to be an absolute constant.
Recall that $W_i = \sum_{j\in [n]} A_{ij}$ is the weighted degree of $i$.

\begin{defn}[$\Delta$-Narrow/Wide Vertex]
  A vertex $i$ is {\em $\Delta$-wide} if the total weight of edges in
  its $\Delta$-prefix is at most $\eta W_i$, and so the weight of
  edges in its $\Delta$-suffix is at least $(1-\eta) W_i$. Otherwise,
  the vertex $i$ is {\em $\Delta$-narrow}.
\end{defn}
Intuitively, a $\Delta$-wide vertex is one where most of its weighted degree
is preserved even if we ignore the $\Delta$ heaviest edges incident to the vertex.

We partition the vertices $V = [n]$ into the \emph{$\Delta$-wide} and
\emph{$\Delta$-narrow} sets; these are respectively denoted $\Vw$ and
$\Vn$.  We define $\Ww := \sum_{i\in \Vw} W_i$ and
$\Wn := \sum_{i\in \Vn} W_i$, and hence the sum of weighted degrees of
all vertices is $W := \sum_{i=1}^n W_i = \Ww + \Wn$.

\begin{defn}[$\Delta$-Narrow/Wide Graph]
  A graph is \emph{$\Delta$-wide} if the sum of weighted degrees of
  $\Delta$-wide vertices accounts for at least $1-\eta$ fraction of
  that of all vertices; i.e., if $\Ww \ge (1-\eta) W$. Otherwise, it is
  $\Delta$-narrow.
\end{defn}


\subsection{Solving \MC for $\Delta$-wide graphs}
\label{sec:high-degree}

Our main theorem for $\Delta$-wide graphs is the following.

\begin{theorem}\label{thm:high-degree}
  Fix $\eps' \in (0,1)$. Given noisy predictions with bias $\eps$,
  there is a polynomial-time randomized algorithm that, given any
  $\Delta$-wide graph, outputs a cut of value at least the maximum cut
  minus 
  $(5\eta + 2\eps') W$, 
  where
  $\Delta := O(1/(\eps \cdot \eps')^2)$, with probability
  $0.98$.

\end{theorem}

Since the graph is $\Delta$-wide, most vertices have their weight
spread over a large number of their neighbors. In this case, the prediction
vector allows us to obtain a good estimate $\hr$ of the optimal
neighborhood imbalance $r^*$ (the difference between the number of
neighbors a vertex has on its side versus the other side of the
optimal cut). We can then write an LP to assign fractional labels to
vertices that maximize the cut value while remaining faithful to these
estimates $\hr$; finally rounding the LP gives the solution.

\subsubsection{The $\Delta$-wide Algorithm}
\label{sec:delta-wide-algorithm}

Define an $n\times n$ matrix $\tA$ from the adjacency 
matrix $A$ as follows: for each row corresponding to the edges incident
to a vertex $i$, we set the entry $\tA_{ij} = 0$ if the edge $(i, j)$
is in the $\Delta$-prefix of vertex $i$; otherwise, $\tA_{ij} = A_{ij}$.
Now, define an $n$-dimensional vector $\hr$ as follows:
\[
  \hr_i = \begin{cases}
    \frac {1}{2\eps} (\tA Y)_i & \text{if $i$ is $\Delta$-wide}\\
    0 & \text{if $i$ is $\Delta$-narrow}
  \end{cases}
\]
where $Y$ is the prediction vector as defined in~(\ref{eq:pred-noisy}).
Solve the linear program: 
\begin{gather}
  \min_{x \in [-1,1]^n} \,\ip{ \hr, x } \quad s.t. \quad \| \hr - Ax \|_1  \leq (\eps' + 2\eta) W.  \label{eq:4}
\end{gather}
Let $\hx \in [-1,1]^n$ be the optimal LP solution.

Finally, do the following $O(\nf1\eta)$ times independently, and
output the best cut $X^*$ among them: randomly round the fractional
solution $\hx$ independently for each vertex to get a cut
$X \in \{-1,1\}^n$; namely, $\Pr[X_i = 1] = \nf{(1+\hx_i)}2$ and
$\Pr[X_i = -1] = \nf{(1-\hx_i)}2$.

\subsubsection{The Analysis}

For a labeling $x \in \{-1,1\}^n$, the \emph{neighborhood imbalance}
for vertex $i$ is defined as $\sum_j A_{ij} x_j = (Ax)_i$. This denotes the 
(signed) difference between the total weight of edges incident to $i$ 
that appear and do not appear in the cut defined by the labeling $x$.
The maximality of the
optimal cut $x^* \in \{-1,1\}^n$ ensures that
$x^*_i \cdot \sign((Ax^*)_i) \leq 0$ for all $i$; else, switching $x_i$
from $1$ to $-1$ or vice-versa
increases the objective.  Define $r^* := Ax^*$ to
be the vector of imbalances for the optimal cut.

\begin{lemma}
  \label{lem:r-values-ok}
  The vector $\hr$ satisfies
  \[
     \EE\left[ \| \hr - r^* \|_1 \right] := \EE\left[ \sum_{i=1}^n
       |\hr_i - r^*_i | \right]
     \leq O\bigg(\frac{W}{\eps \sqrt{\Delta}}\bigg) + 2\eta W.
  \]
\end{lemma}

\begin{proof}
    Observe that
    \[
        \EE[Y_i] 
        = x^*_i\cdot \Pr[Y_i = x^*_i] - x^*_i\cdot  \Pr[Y_i = -x^*_i]
        = x^*_i (\nf 12 + \eps) - x^*_i (\nf 12 - \eps)
        = 2\eps x^*_i.
    \]
  Define $Z := \frac{1}{2\eps}Y$. Then, $\EE[Z] = x^*$, and so
  $\EE[AZ] = r^*$. 

  First, we consider a $\Delta$-narrow vertex $i$. Since $\hr_i = 0$, we have
  $|\hr_i - r^*_i| = |r^*_i| \le W_i$.  
  So summing over all $\Delta$-narrow vertices gives
  \begin{equation}\label{eq:narrow}
    \sum_{i\in \Vn} |\hr_i - r^*_i| 
    \le \sum_{i\in \Vn} W_i
    \le \eta W,
  \end{equation}
  since the graph is $\Delta$-wide.

  Now, we consider a $\Delta$-wide vertex $i$.  We have
  \begin{equation}\label{eq:wide}
    |\hr_i - r^*_i| = |(\tA Z)_i - r^*_i| \le |\EE[(\tA Z)_i] - r^*_i| + |(\tA Z)_i - \EE[(\tA Z)_i]|.
  \end{equation}

  To bound the first term in the RHS of \eqref{eq:wide}, recall that $r^*_i = \EE[(AZ)_i]$. Thus,
  \[
      |\EE[(\tA Z)_i] - r^*_i|
      = |\EE[(\tA Z)_i] - \EE[(AZ)_i]|
      = \ip{(\tA - A)_i, \EE[Z]}.
  \]
  Since $\EE[Z] = x^*\in \{-1, 1\}^n$, we get 
  \[
    |\EE[(\tA Z)_i] - r^*_i|
    = (\tA - A)_i \cdot x^*
    \le \|(\tA - A)_i\|_1 \|x^*\|_\infty
    \le \eta W_i,
  \]
  where in the last step, we used the fact that $i$ is a $\Delta$-wide vertex.

\eat{
  Now, we bound the second term in the RHS of \eqref{eq:wide}.   Using
  a Hoeffding bound on the sum $\sum_j \tA_{ij}Z_j$, we get
  \[ \Pr[|(\tA Z)_i - \EE[(\tA Z)_i ]| \ge \lambda_i ] \leq \exp\bigg(-
    \frac{O(\lambda_i^2)}{\sum_{j\in [n]} (\tA_{ij}/\eps)^2}\bigg). \]
  We know
  \[
    \sum_{j\in [n]} \tA_{ij}^2 = \|\tA_i\|_2^2 \le \|\tA_i\|_1 \cdot \|\tA_i\|_\infty.
  \]
  Note that the weight of any edge in the $\Delta$-suffix of $i$ is at most $W_i/\Delta$.
  Therefore, by our definition of $\tA$, we have $\|\tA_i\|_\infty \le W_i/\Delta$.
  Since $\tA_{ij} \le A_{ij}$ for all $j\in [n]$, we also have $\|\tA_i\|_1 \le  \|A_i\|_1 = W_i$.
  Applying these bounds, We get:
  \[
    \sum_{j\in [n]} \tA_{ij}^2 \le W_i^2/\Delta.
  \]
  Therefore,
  \[ \Pr[|(\tA Z)_i - \EE[(\tA Z)_i ]| \ge \lambda_i ] \leq
    \exp\bigg(- \frac{O(\lambda_i^2) \cdot
      \Delta}{(W_i/\eps)^2}\bigg). \] Setting
  $\lambda_i := O({W_i}/(\eps \sqrt{\Delta}))$ makes
  the RHS at most $\nf12$. Since the RHS will at least exponentially
  decrease when the deviation bound becomes
  $2\lambda_i, 3\lambda_i, \dots$, so
  $\E[|(\tA Z)_i - \E[(\tA Z)_i|] \leq O(\lambda_i)$.
  Plugging back the bounds on the two terms in the RHS of \eqref{eq:wide}, we get
  for a $\Delta$-wide vertex $i$, 
  \[
    \EE[|\hr_i - r^*_i|] \le O\left(\frac{W_i}{\eps \sqrt{\Delta}}\right) + \eta W_i.
  \]
}

Now, we bound the second term in the RHS of \eqref{eq:wide}. Using
  Chebyshev's inequality on the sum $(\tA Z)_i = \sum_j \tA_{ij}Z_j$,
  we get
  \[ \Pr[|(\tA Z)_i - \EE[(\tA Z)_i ]| \ge \lambda_i ] \leq
    \frac{\text{var}((\tA Z)_i)}{\lambda_i^2}. \] Since the variables
  $Z_j$ are pairwise independent, the variance
  $\text{var}((\tA Z)_i) = \sum_j \tA^2_{ij} \text{var}(Z_j)$. The
  individual variances are at most $1$, hence we need to bound $\sum_j
  \tA^2_{ij}$. 
  We know
  \[
    \sum_{j\in [n]} \tA_{ij}^2 = \|\tA_i\|_2^2 \le \|\tA_i\|_1 \cdot \|\tA_i\|_\infty.
  \]
  Note that the weight of any edge in the $\Delta$-suffix of $i$ is at most $W_i/\Delta$.
  Therefore, by our definition of $\tA$, we have $\|\tA_i\|_\infty \le W_i/\Delta$.
  Since $\tA_{ij} \le A_{ij}$ for all $j\in [n]$, we also have $\|\tA_i\|_1 \le  \|A_i\|_1 = W_i$.
  Applying these bounds, We get:
  \[
    \sum_{j\in [n]} \tA_{ij}^2 \le W_i^2/\Delta.
  \]
  Therefore,
    \[ \Pr[|(\tA Z)_i - \EE[(\tA Z)_i ]| \ge \lambda_i ] \leq
      \frac{(W_i/\eps)^2}{O(\lambda_i^2) \cdot
        \Delta}. \] Setting
    $\lambda_i := 2^t \cdot O({W_i}/(\eps \sqrt{\Delta}))$ makes the right hand
    side at most $\frac{1}{2^{2t}}$, say. 
    Now, for a non-negative random variable $Z$ and some scalar $z$,
    we have \[ \E[Z] \leq z + \sum_{t\geq 0} \Pr[ Z \geq 2^t\cdot z]
    \cdot 2^{t+1}z. \] Using this above, we get that
    \[ \E[|(\tA Z)_i - \EE[(\tA Z)_i ]|] = O({W_i}/(\eps
      \sqrt{\Delta})). \]
  
  Summing over all $\Delta$-wide vertices, we get
  \[
    \EE\left[\sum_{i\in \Vw} |\hr_i - r^*_i|\right]
    \le O\left(\frac{\Ww}{\eps \sqrt{\Delta}}\right) + \eta \Ww
    \le O\left(\frac{W}{\eps \sqrt{\Delta}}\right) + \eta W.
  \]
  Combining with \eqref{eq:narrow} for $\Delta$-narrow vertices, we get
  \[
    \EE\left[\|\hr_i - r^*_i\|_1\right]
    \le O\left(\frac{W}{\eps \sqrt{\Delta}}\right) + 2\eta W.\qedhere
  \]
\end{proof}

Now using Markov's inequality with \Cref{lem:r-values-ok}, we get that
setting $\Delta = \Omega(1/(\eps \eps')^2)$ for any fixed constant
$\eps' > 0$ ensures that we get a vector of empirical imbalances $\hr$
satisfying
\begin{gather}
  \| \hr - r^* \|_1 \leq (\eps'+2\eta) W. \label{eq:5}
\end{gather}
with probability at least $0.99$. (Since the $2\eta W$ losses are
deterministically bounded, we can use Markov's inequality only on the
random variable $\sum_{i \in \Vw} |(\tA Z)_i - \E[(\tA Z)_i]|$.)
Hence, when the event in (\ref{eq:5}) happens, the vector $x^*$ is a
feasible solution to LP~(\ref{eq:4}).

Next, we need to analyze the quality of the cut produced by randomly
rounding the solution of LP~(\ref{eq:4}). Recall that for the
(unnormalized) Laplacian $L$ and some $x \in \{-1,1\}^n$, the cut
value is
\begin{gather}
  f(x) := \nf14 \cdot \ip{x, Lx} = \nf14 \cdot ( \ip{ x, Dx } - \ip{ x, A x}) = \nf14 \cdot ( W - \ip{ x, A x}). \label{eq:1}
\end{gather}

\begin{lemma}
  \label{lem:algo-value}
  For any $\Delta$-wide graph, the algorithm
  from \S\ref{sec:delta-wide-algorithm} outputs 
  $X^* \in \{-1,1\}^n$ that satisfies
    \[ f(X^*) \ge f(x^*) - (2\eps'+5\eta)W\]
  with probability at least $0.98$.
\end{lemma}

\begin{proof}
  Recall that the cut $X^*$ is the best among
  $T := O(\nf1\eta)$ independent roundings of cut
  $\hx$. Consider one of the roundings $X$, and write:
  \begin{gather}\label{eq:bound}
    \ip{X,AX} = \ip{\hx, \hr} + (\ip{\hx, A\hx} - \ip{\hx, \hr}) +
    (\ip{X, AX} - \ip{\hx, A\hx}).
  \end{gather}
  Let us first bound the expectation of each of the terms in the RHS
  of \eqref{eq:bound} separately.

  To bound the first term $\ip{\hx, \hr}$, note that given
  \eqref{eq:5} (which happens with probability $0.99$), the solution $x^*$ is feasible for the LP in
  \eqref{eq:4}. This means the optimal solution $\hx$ has objective function value
  \begin{align}
   \ip{\hr,\hx} &\leq \ip{ \hr, x^* } = \ip{ r^*, x^*} + \ip{ \hr -
      r^*, x^*} \leq \ip{x^*, Ax^*} + \| \hr - r^* \|_1 \| x^*
    \|_\infty \notag \\ &\leq \ip{x^*, Ax^*} + (\eps'+2\eta) W. \label{eq:6}
  \end{align}
  Next, we bound the second term $(\ip{\hx, A\hx} - \ip{\hx, \hr})$ by
  \begin{gather}
    \| \hx \|_{\infty} \| A\hx - \hr \|_1 \leq (\eps' + 2\eta) W, \label{eq:3}
  \end{gather}
  by feasibility of $\hx$ for the LP in \eqref{eq:4}.
  Finally, we bound the third term $(\ip{X, AX} - \ip{\hx, A\hx})$,
  this time in expectation:
  \begin{gather}
    \E[\ip{X, AX}] - \ip{\hx, A\hx} = 0. \label{eq:2}
  \end{gather}
  Chaining \cref{eq:6,eq:3,eq:2} for the various parts of~(\ref{eq:bound}),
  we get 
  \[
    \E[\ip{X, AX}] \leq \ip{x^*, Ax^*} + (2\eps'+4\eta)W.
\]

    Moreover, using that $\ip{X, AX} \in [-W,W]$, we get
  \begin{align*}
    \Pr\big[\ip{X, AX} \geq \E[\ip{X, AX}] + \eta W\big] &=
    \Pr\big[\ip{X, AX} + W \geq \E[\ip{X, AX}] + (1+\eta) W\big]\\ 
    &\leq  
    \Pr\big[\ip{X, AX} + W \geq (1+\eta/2)\,\left(\E[\ip{X, AX}] + W\big]\right)\\
    &\leq \nicefrac{1}{(1+\eta/2)}.
  \end{align*}
  If $X^*$ is the cut with the smallest value of $\ip{X, AX}$ among
  all the independent roundings:
  \[ \Pr\big[\ip{X^*, AX^*} \leq \ip{x^*, Ax^*} + (2\eps'+5\eta)W \big]
    \geq 1 - (\nicefrac{1}{(1+\eta/2)})^T \geq 0.99. \]
  Substituting into the definition of $f(\cdot)$ completes the proof.

\end{proof}

This proves \Cref{lem:algo-value}, and hence also
\Cref{thm:high-degree}.

\emph{Deterministic Rounding.} We observe that we can replace the
repetition by a simple pipeage rounding algorithm to round the
fractional solution $\hx$ to an integer solution $X^*$ without
suffering any additional loss. Indeed, viewing $\ip{x, Ax}$ as a
function of some $x_i$ keeping the remaining
$\{x_1, \ldots, x_n\} \setminus \{x_i\}$ fixed gives us a linear
function of $x_i$ (since the diagonals of $A$ are zero). Hence we can
increase or decrease the value of $x_i$ to decrease $\ip{x,Ax}$ until
$x_i \in \{-1,1\}$. Iterating over the variables gives the
result. However, this does not change the result qualitatively.

\subsection{Solving \MC for $\Delta$-narrow graphs}
\label{sec:low-degree}

Next, we consider $\Delta$-narrow graphs. Our main result for them is the following.    

\begin{theorem}\label{thm:narrow}
    For any $\Delta \in \N$, there is a randomized 
    algorithm for the \MC problem with an (expected) approximation factor of 
    $\alphaGW + \tilde{\Omega}(\eta^5/\Delta^2)$ on any $\Delta$-narrow graph.
\end{theorem}


For the case of $\Delta$-narrow graphs, we do not use predictions;
rather, we adapt an existing algorithm for the \MC problem for
low-degree graphs by Feige, Karpinski, and Langberg~\cite{FeigeKL02}
and its refinement due to Hsieh and Kothari~\cite{HsiehK22}.  Note
that any graph with maximum degree $\Delta$ is clearly $\Delta$-narrow
(even when $\eta = 1$).

\subsubsection{The $\Delta$-narrow Algorithm}
\label{sec:delta-narrow-algorithm}

We show that \Cref{thm:narrow} holds for the Feige, Karpinski, and
Langberg (FKL) \MC algorithm~\cite{FeigeKL02}. We briefly recall this
algorithm first. Consider the \MC SDP with triangle inequalities:
\[
    \max_{v_i\in S_n~\forall i\in [n]} \sum_{i, j\in [n]: A_{ij} = 1} \frac{1-\ip{v_i, v_j}}{2} \quad s.t. \quad s( \ip{v_j, v_k} + \ip{v_i, v_k}) \le \ip{v_i, v_j} \quad \forall i, j, k \in [n], s\in \{-1, 1\}.
\]
where $S_n$ is the unit sphere of $n$ dimensions. Let $\hv$ be an optimal solution to this SDP.

Let $g$ be a unit vector chosen uniformly at random from $S_n$. We use
\emph{random hyperplane rounding} (cf.\ the Goemans-Williamson \MC algorithm~\cite{GoemansW95}) to round $\hv$ to $\hx\in \{-1, 1\}^n$ as follows: if $\ip{\hv_i, g} > 0$, then $\hx_i = 1$; else, $\hx_i = -1$. 

Now, define $F = \{i\in [n]: \ip{\hv_i, g} \in [-\delta, \delta]  \}$ for some $\delta = \Theta(1/((\Delta/\eta)\sqrt{\log (\Delta/\eta)}))$. 
We partition $N_i := [n] \setminus \{ i \}$ as follows: $B_i := \{j\in N_i\setminus F: \hx_j = \hx_i\}$, and $C_i := \{j\in N_i\setminus F: \hx_j \not= \hx_i\}$ and $D_i := N_i\cap F$. We define $F'\subseteq F$ as follows: $i\in F'$ if $i\in F$ and $w(B_i) > w(C_i)+w(D_i)$ where $w(S) := \sum_{j \in S} A_{ij}$. In the final output $X\in \{-1, 1\}^n$, we flip the vertices in $F'$, namely $X_i = -\hx_i$ if $i\in F'$, else $X_i = \hx_i$.

\subsubsection{The Analysis}

The ``local gain'' for a vertex $i\in F$ is defined as $\Delta_i := (|B_i| - (|A_i| + |C_i|))^+$, where $z^+ = \max(z, 0)$.
We now state the following key lemmas: 

\begin{lemma}\label{lem:hk1}
     For any vertex $i\in [n]$, $\Pr[i\in F] = \Omega(\delta)$.
\end{lemma}
\begin{proof}
    This lemma immediately follows from \cite[Fact 3]{HsiehK22}.
\end{proof}

Let $i$ be a $\Delta$-narrow vertex, and $w \in \R^{n}$ be its weight vector ($w_i = A_{ij}$ for all $j \in [n]$) so that $W_i = \| w_i \|_1$. 
Let $w' \in \R^n$ be the projection of $w$ onto its top $\Delta$ coordinates. 
The narrowness of $i$ implies that $\| w' \|_1 \geq \eta \| w \|_1$, which implies that 
\[
\| w \|_2^2 \geq 
\| w' \|_2^2 \geq \frac{ \| w' \|_1^2}{\Delta} \geq 
\frac{\eta^2 \| w \|_1^2}{ \Delta}.
\]

It turns out that the analysis of~\cite{HsiehK22} still holds under the above bound between $\ell_1$ and $\ell_2$ norms of weight vectors. So we have the following slight generalization of their Lemma 8. 

\begin{lemma}[extends Lemma 8 of \cite{HsiehK22}]\label{lem:hk2}
    There is a large enough constant $C$ such that for any $d \geq 3$ and $\delta = \frac{1}{Cd\sqrt{\log d}}$,  
    for any vertex $i$ whose weight vector $w$ satisfies $\| w \|_1^2 \leq d \| w \|_2^2$, it holds that the expected local gain of a vertex $i$ satisfies:
    \[
        \EE[\Delta_i|i\in F] = \Omega\left(\frac{W_i}{d\sqrt{\log d}}\right).
    \]
\end{lemma}
\begin{proof}
    In~\cite{HsiehK22}, the only place where the degree bound $d$ is used is $\| w \|_1^2 \leq d \| w \|_2^2$ at the end of the proof of Lemma 7.
\end{proof}

We are now ready to prove that the FKL algorithm establishes \Cref{thm:narrow}.

\begin{proof}[Proof of \Cref{thm:narrow}]
    Note that the value of the cut $X$ exceeds that of $\hx$ by $\sum_{i\in F'} \Delta_i$, i.e.,
    \begin{align*}
        \EE[\ip{X, LX}] 
        &=  \EE[\ip{\hx, L\hx}] + \sum_{i\in [n]} \EE[\Delta_i|i\in F]
        \cdot \Pr[i\in F] \\
        &\ge \EE[\ip{\hx, L\hx}] + \sum_{i: \Delta\mbox{-narrow}} \EE[\Delta_i|i\in F] \cdot \Pr[i\in F].
      \end{align*}
      
    Let the approximation factor of the cut $\hx$ output by the Goemans-Williamson algorithm be denoted $\alphaGW$
    and let $\opt$ be the size of the maximum cut. Then,
    \[
        \EE[\ip{\hx, L\hx}] \ge \alphaGW \cdot \opt.
    \]
    From \Cref{lem:hk1,lem:hk2} with $d = \Delta/\eta^2$, we get 
    \[
        \EE[\ip{X, LX}] 
        \ge 
        \alphaGW\cdot \opt + \Omega\left(
        \frac{1}{(\Delta/\eta^2) \sqrt{\log (\Delta/\eta^2)}} \cdot
        \sum_{i: \Delta\mbox{-narrow}} \frac{W_i}{(\Delta/\eta^2)\sqrt{\log (\Delta/\eta^2)}}\right).
    \]
    Since $\sum_{i: \Delta\mbox{-narrow}}W_i \ge \eta W \geq 2 \eta \cdot \opt$, we get 
    \[
        \EE[\ip{X, LX}] 
        \ge (\alphaGW + \tilde{\Omega}(\eta^5/\Delta^2) ) \cdot \opt.\qedhere
    \]
\end{proof}

\subsection{Wrapping up: Proof of \Cref{thm:gen-degree}}

For $\Delta$-wide graphs, \Cref{thm:high-degree} returns a cut with
value at least
\[ \opt - (2\eta + \eps')W \] 
with probability $0.98$. Since we can
always find a cut of value $\alphaGW \cdot \opt$, and $\opt \geq W/2$, this
means the expected cut value is at least
\[ \big[ 0.98 \cdot (1 - 6\eta - 2\eps') + 0.02 \cdot \alphaGW \big]
  \opt. \] And for $\Delta$-narrow graphs, \Cref{thm:narrow} finds
a cut with expected value
\[ \big[ \alphaGW + \tilde\Omega(\eta^5/\Delta^2)\big] \cdot \opt. \]
Moreover, recall that $\Delta = O(1/(\eps\eps')^2)$. Setting 
$\eta, \eps'$ to be suitably small universal constants gives us that
both the above approximation factors are at least $\alphaGW +
\tilde{\Omega}(\eps^4)$, which proves~\Cref{thm:gen-degree}.

\eat{
\subsection{Putting things together: Arbitrary Degrees}
Given $\eps$, we set the threshold $d = \Omega(1/\eps^2)$. If $\tau(d) \geq 0.01$, then we run the FKL algorithm and by \Cref{thm:low-degree}, the approximation factor is $\alpha_{\rm GW} + \tilde{\Omega}(1/d^2) = \alpha_{\rm GW} + \tilde{O\Omega}(\eps^4)$. 

Otherwise, we remove all vertices of degree at most $d$ simultaneously to form a new graph $H$. Since $\tau(d) \le 0.01$, this removes at most $0.01 m$ edges. Let $X$ be the set of vertices in $H$.  Since every vertex $i\in X$ had degree at least $d$ in the original
graph, $|X| \le 2m/d$. Let $d^H_i$ denote the degree of vertex $i\in X$ in the new graph $H$. Then, the average degree of vertices in $H$ is given by
\[
    \frac {1}{|X|}\cdot  \sum_{i\in X} d^H_i \ge \frac{d}{2m}\cdot  (2m - 2\cdot 0.01 m) = 0.99 d,
\]
where we used the facts that at most $0.01 m$ edges were removed and $|X| \le 2m/d$. 

Next, we sequentially remove vertices of degree less than $0.01 \cdot 0.99d$. The total number of edges that can be removed is at most $0.01\cdot 0.99d \cdot |X|\le 0.01\cdot 0.99d \cdot (2m/d) \le 0.02 m$. This ensures that the resulting graph has minimum degree at least $\Omega(d)$ while still retaining at least $0.97 m$ edges. If $d = \Omega(1/\eps^2)$ so that we can apply \Cref{thm:high-degree} with $\eps' = 0.01$, we can get a cut of size at least $\opt - 0.04 m \geq 0.96\cdot \opt$. 

This completes the proof of \Cref{thm:gen-degree}.
}





\section{MaxCut in the Partial Prediction Model}
\label{sec:max-cut-partial}

We now consider the partial prediction model, where each vertex pairwise-independently reveals their correct label with probability $\eps$. Since an $\eps^2$ fraction of the edges are induced by the vertices with the given labels, it is easy
to get an approximation ratio of almost $\alphaGW + \Omega(\eps^2)$.

\begin{theorem}\label{thm:partial-info-gw}
    Given noisy predictions with a rate of $\eps$,
    there is a polynomial-time randomized algorithm that obtains 
    an (expected) approximation factor of $\alphaGW +\eps^2$
    for the \MC problem
\end{theorem}
\begin{proof}
Given a graph $G = (V, E)$ with the optimal cut $(A^*, B^*)$ that cuts $E^* = E \cap E(A^*, B^*)$, let $S$ be the set of vertices whose label is given, and let $A = A^* \cup S$, $B = B^* \cup S$. 
Consider the following \MC SDP that fixes the vertices with the revealed labels. 
\begin{align*}
\max_{v_i\in S_n~\forall i\in [n]} & \sum_{i, j\in [n]} \frac{A_{i,j}(1-\ip{v_i, v_j)}}{2} \quad s.t.\ v_i = v_0 \ \forall i \in A\mbox{ and } v_i = -v_0 \ \forall i \in B.
\end{align*}
Note that this is still a valid relaxation so the optimal SDP value $\sdp$ is at least $\opt$.
For each edge $e \in E^*$, $e \in E(A, B)$ with probability $\eps^2$; in other words, both of its endpoints will reveal their labels. Let $\tau$ denotes the total weight of such edges, so that $\E[\tau] = \eps^2 \opt$. Note that $\sdp \geq \opt$ for every partial prediction. 

Perform the standard hyperplane rounding. For each $e \in E^* \cap E(A, B)$, the rounding will always cut $e$. For all other edges, we have an approximation ratio of $\alphaGW$. Therefore, the expected weight of the cut edges is at least 
\[
\E[\tau W + \alphaGW(\sdp - \tau W)]  
\geq \eps^2 \opt + \alphaGW(1 - \eps^2)\opt =
(\alphaGW + (1 - \alphaGW)\eps^2)\cdot \opt. \qedhere
\]
\end{proof}

Is $\Omega(\eps^2)$ optimal? Ideally, we could get an $\Omega(\eps)$-advantage if the hyperplane rounding performs better than $\alphaGW$ for the edges with only one endpoint's label revealed. One naive way to achieve this is to hope that the rounding {\em preserves the marginals}; i.e., $\E[x_i] = \langle v_0, v_i \rangle$ for all $i \in [n]$. In that case, if we consider $(i, j)$ where if $v_i = \pm v_0$, the probability that $(i, j)$ is cut is exactly their contribution to the SDP $(1 - \langle v_i, v_j \rangle)$. 

Unfortunately, the hyperplane rounding does not satisfy this property. Instead, we use the rounding scheme developed by Raghavendra and Tan~\cite{RT12} for max-bisection that has an approximation ratio $\alphaRT \approx 0.858$ while preserving the marginals. 


\PartialThm*

\begin{proof}
Given a graph $G = (V, E)$ with the optimal cut $(A^*, B^*)$ that cuts $E^* = E \cap E(A^*, B^*)$, let $S$ be the set of vertices whose label is given, and let $A = A^* \cup S$, $B = B^* \cup S$. 
Let $E'$ be the set of the edges that are incident on $A$ \underline{\emph{or}} $B$. 
Each edge cut in the optimal solution will be in $E'$ with probability $2\eps - \eps^2$. Let $\tau$ be the total weight of the edges in $E^* \cap E'$ so that $\E[\tau] = (2\eps - \eps^2)\opt$. Guess the value of $\tau$ (up to a $o(1)$ multiplicative error that we will ignore in the proof), and consider the following \MC SDP that fixes the vertices with the revealed labels and requires a large SDP contribution from $E'$. 
\begin{align*}
\max_{v_i\in S_n~\forall i\in [n]}\ & \sum_{i, j\in [n]} \frac{A_{i,j}(1-\ip{v_i, v_j)}}{2} \\
s.t.\ & v_i = v_0 \ && \forall i \in A \\
& v_i = -v_0 \ && \forall i \in B  \\
& \sum_{(i, j) \in E'} \frac{A_{i, j}(1-\ip{v_i, v_j)}}{2} \geq \tau. 
\end{align*}

Given the correctly guessed value of $\tau$,  the optimal solution is still feasible for the above SDP, so $\sdp \geq \opt$. 
We use Raghavendra and Tan~\cite{RT12}'s rounding algorithm, which is briefly recalled below. 

\begin{itemize}
    \item For each $i \in [n]$, define $\mu_i \in [-1, +1]$ and $w_i \in \mathbb{R}^n$ such that $v_i = \mu_i v_0 + w_i$ and $w_i \perp v_0$. Let $\overline{w_i} = w_i / \| w_i \|$. ($w_i = 0$ if and only if $v_i = \pm v_0$. Then define $\overline{w_i} = 0$.)
    
    \item Pick a random Gaussian vector $g$ orthogonal to $v_0$. Let $\xi_i := \langle g, \overline{w_i} \rangle$. Note that each $\xi_i$ is a standard Gaussian.

    \item Let the threshold $t_i := \Phi^{-1}(\mu_i/2 + 1/2)$ where $\Phi$ is the CDF of a standard Gaussian. 

    \item If $\xi_i \leq t_i$, set $x_i = 1$ and otherwise set $x_i = -1$.
\end{itemize}
    Raghavendra and Tan showed that this rounding achieves an $(\alphaRT \approx 0.858)$-approximation for every edge. 
    
    Consider an edge $(i, j) \in E'$ and without loss of generality, assume $i \in B$, which implies that $v_i = -v_0$. The contribution of this edge to the SDP objective is $\mu_j/2 + 1/2$. Note that $\Pr[x_j = 1]$ is exactly $\mu_j/2 + 1/2$ and $\E[x_j] = (\mu_j/2 + 1/2) - (1/2 - \mu_j/2) = \mu_j$.
    So, we get a $1$-approximation from this edge. Since other edges still have an $\alphaRT$-approximation, the total expected weight of the edges cut is at least 
    \[
    \E[\tau + 
    \alphaRT(\sdp - \tau)]
    \geq (2\eps - \eps^2) \opt + \alphaRT (1 - (2\eps - \eps^2)) \opt 
    = \alphaRT \cdot \opt + 
    (1 - \alphaRT)(2\eps - \eps^2)\opt.
  \]
  Hence the proof of~\Cref{thm:partial-info-rt}.
\end{proof}



\section{2-CSPs in the Noisy Prediction Model}\label{section:csps}
In this section we extend \cref{thm:high-degree} to general $2$-CSPs.

We start by introducing some additional notation and definitions.
For a multi-index $\alpha\in [n]^2$ we denote by  $\alpha(i)$ its $i$-th index. For variables $x_1,\ldots,x_n\,,$ we then write $\chi_\alpha(x)$ for the monomial $\prod_{i\in \alpha}x_i\,.$
Given a predicate $P:\{-1,+1\}^2\rightarrow\{0,1\}\,,$ an instance $\cI$ of the CSP(P) problem over variables $x_1,\ldots, x_n$ is a multi-set of triplets $(w,c,\alpha)$ representing constraints of the form $P(c\circ x^\alpha)=P(c_1x_{\alpha(1)},c_2x_{\alpha(2)})=1$ where $\alpha\in [n]^2$ is the scope, $c\in \{\pm 1\}^2$ is the negation pattern and $w\geq 0$ is the weight of the constraint. We let $W=\sum_{(w,c,\alpha)\in \cI}w\,.$ We can represent the predicate $P$ as the multilinear polynomial of degree $2$ in indeterminates $x_{\alpha(1)},x_{\alpha(2)}\,,$
\begin{align*}
    P(c\circ x^\alpha)=\sum_{\alpha'\subseteq \alpha} c^\alpha \cdot \hat{p}(\alpha')\cdot \chi_{\alpha'}(x)\,,
\end{align*}
where $\hat{p}(\alpha')$ is the coefficient in $P$ of the monomial $\chi_{\alpha'}(x)\,.$ Notice that this formulation does not rule out predicates with same multi-index but different negation
pattern or multi-indices in which an index appears multiple times.
Given a predicate $P$, an instance $\cI$ of CSP(P) with $m$ constraints and $x\in \set{\pm 1}^n$ we define 
\begin{align*}
    \val_\cI(x):=\frac{1}{W}\sum_{(w,c, \alpha)\in \cI} w\cdot P(c\circ x^{\alpha})\qquad\textnormal{and}\qquad
    \opt_\cI:=\max_{x\in \set{\pm 1}^n} \val_{\cI}(x)\,.
\end{align*}
For an instance $\cI$ of CSP(P), in the noisy prediction model we assume there is some fixed assigment  $x^*$ with value $\val_\cI(x^*)=\opt_\cI\,.$ The algorithm has access to a prediction vector $Y\in \set{\pm 1}^n$ such that predictions $y_i$'s are $2$-wise independently correct  with probability $\frac{1+\eps}{2}$ for unknown bias $\eps\,.$ 
We let $Z=\frac{Y}{2\eps}\,.$ With a slight abuse of notation we also write $P(c\circ Z^{\alpha})$ even though $Z$ is a rescaled boolean vector.

For a literal $i\in [n]$ and an instance $\cI$ of CSP(P) we let $S_i:=\set{(w,c,\alpha)\in \cI \,|\, \alpha(1)=i}\,.$ 
As in \cref{sec:delta-wide-algorithm}, we can define $\Delta$-wide literals and instances. For an instance $\cI\,,$ we partition the constraints in $S_i$ into two sets: the {\em $\Delta$-prefix} for
$i$ comprises the $\Delta$ heaviest constraints in $S_i$ (breaking
ties arbitrarily), while the remaining constraints make up the {\em
  $\Delta$-suffix} for $i$, which we denote by $\tilde{S}_i$. We fix a parameter $\eta \in
(0,\nf12)$. 
We let $W_i=\sum_{(w,c,\alpha)\in \cS_i}w_i\,.$

\begin{defn}[$\Delta$-Narrow/Wide]
  A literal $i$ is {\em $\Delta$-wide} if the total weight of its in
  its $\Delta$-prefix is at most $\eta W_i$, and so the weight of
  edges in its $\Delta$-suffix is at least $(1-\eta) W_i$. Otherwise,
  the literal $i$ is {\em $\Delta$-narrow}. An instance $\cI$ of CSP(P) is \emph{$\Delta$-wide} if $\sum_{\substack{i \in [n]\\\Delta\textnormal{-wide}}} W_i\ge (1-\eta) W$.
\end{defn}

We are now ready to state the main theorem of the section.

\begin{theorem}\label{thm:csps-high-degree}
    Let $P:\set{\pm 1}^2\rightarrow \set{0,1}$ be a predicate.
    Let $\eps' \in (0,1)\,,$ $\eta\in (0,1/2)$ and $\Delta\geq O(1/(\eps'\cdot\eps)^2)$.
    There exists a polynomial-time randomized algorithm that, given a $\Delta$-wide $\cI$ in CSP(P) and noisy predictions with bias $\eps$,
    returns a vector $\hat{x}\in \set{\pm 1}^n$ satisfying
    \begin{align*}
        \val_\cI(x)\geq \opt_\cI - 5\eta -O\Paren{\eps'}\,,
    \end{align*}
    with probability at least $0.98\,.$
\end{theorem}

The proof of \cref{thm:csps-high-degree} follows closely that of \cref{thm:high-degree}. 
First observe that we may assume without loss of generality that each $(w,c,\alpha)$ appears exactly twice in $\cI\,.$
This is convenient so that for all $x\in \set{\pm 1}^n\,,$ $\val_\cI(x)=\sum_{i\in [n]}\sum_{(w,c,\alpha)\in S_i} w\cdot P(c\circ x^\alpha)\,.$
With a slight abuse of notation, for all $(w,c,\alpha)\in S_i\,,$ we let
\begin{align*}
    P(c\circ (x_i\cdot Z^{\alpha\setminus i})):=\sum_{\substack{\alpha'\subseteq \alpha\textnormal{ s.t.}\\ \alpha'(1)=i}} \hat{p}_{\alpha'}c^{\alpha'} x_i\cdot \chi_{\alpha'\setminus \alpha'(1)}(Z) + \sum_{\substack{\alpha'\subseteq \alpha\textnormal{ s.t.}\\\alpha'(1)\neq i}} \hat{p}_{\alpha'}c^{\alpha'} \chi_{\alpha'}(Z)\,,
\end{align*}
and
\begin{align*}
    P(c\circ x^{\alpha \setminus i}):=\sum_{\substack{\alpha'\subseteq \alpha\textnormal{ s.t.}\\ \alpha'(1)=i}} \hat{p}_{\alpha'}c^{\alpha'}  \chi_{\alpha'\setminus \alpha'(1)}(x) + \sum_{\substack{\alpha'\subseteq \alpha\textnormal{ s.t.}\\\alpha'(1)\neq i}} \hat{p}_{\alpha'}c^{\alpha'} \chi_{\alpha'}(x)\,,
\end{align*}
We further define $\tilde{S}_i\subseteq S_i$ to be subset of constraints in $S_i$ that are not part of the $\Delta$-prefix of $i\,.$
We can now state the algorithm behind \cref{thm:csps-high-degree}, which amounts to the following two steps.
\begin{enumerate}
    \item Solve the linear program 
    \begin{align*}
        \max_{x\in [-1,+1]^n}\sum_{\substack{i\in [n]\\ \Delta\textnormal{-wide}}}\sum_{(w,c,\alpha)\in \tilde{S}_i} w P(c\circ (x_i\cdot Z^{\alpha\setminus i}))
    \end{align*}
    subject to
    \begin{align}
        \sum_{\substack{i\in[n]\\\Delta\textnormal{-narrow}}}&\left|\sum_{(w,c,\alpha)\in S_i}w P(c\circ x^{\alpha \setminus i})\right|+\sum_{\substack{i\in[n]\\\Delta\textnormal{-wide}}}\left|\sum_{(w,c,\alpha)\in S_i\setminus \tilde{S}_i}w P(c\circ x^{\alpha \setminus i})\right|\nonumber\\
        +\sum_{\substack{i\in[n]\\\Delta\textnormal{-wide}}}&\left| 
        \sum_{(w,c,\alpha)\in \tilde{S}_i}w \Paren{P(c\circ x^{\alpha \setminus i})-P(c\circ Z^{\alpha \setminus i})}\right|\leq  C (\eps'+2\eta)W\label{eq:csp-lp-constraint}
    \end{align}
    for some large enough absolute constant $C>0\,.$ Let $\hat{x}\in\brac{-1,+1}^n$ be the found optimal solution.
    \item Repeat $O(1/\eta)$ times independently and output the best assignment $X^*\,:$ independently for each $i\in [n]$ set $X_i=1$ with probability $(1+\hat{x}_i)/2$ and $X_i=-1$ otherwise.
\end{enumerate}

The LP above generalize the one in \cref{eq:4}, which comes as a special case where $\hat{p}_{\alpha'}=0$ for all $\alpha'\subset \alpha\in[n]^2\,.$ Indeed, since predicates contain only two literals, the program is  linear. 
Given the resemblance between \cref{eq:4} and \cref{eq:csp-lp-constraint}, the proof of \cref{thm:csps-high-degree} follows closely that of \cref{thm:high-degree}, we defer it to \cref{sec:missing-proofs}.

\section{Closing Remarks}

Our work suggests many directions for future research. One immediate
question is to quantitatively improve our results: e.g., the exponent
of $\eps$ for noisy predictions, and the constants. Here are some
broader questions:

\begin{enumerate}

\item 
  We assume that our
  noisy predictions are correct with probability \emph{equal to} $\nf12 + \eps$; we
  can easily extend to the case where each node has a prediction that
  is correct with some probability $\nf12 + \eps_i$, and each $\eps_i
  \in \Theta(\eps)$. But our approach breaks down when different nodes
  are correct with wildly different probabilities, even when we are guaranteed $\eps_i \geq \eps$ for every $i$. Can we extend to
  that case?  

\item For which other problems can we improve the performance of the
  state-of-the-art algorithms using noisy predictions? As we showed, the
  ideas used for the $\Delta$-wide case extend to more general
  maximization problems on $2$-CSPs with ``high degree'', but can we
  extend the results for the ``low-degree'' case where each variable
  does not have a high-enough degree to infer a clear signal? Can we
  extend these to minimization versions of $2$-CSPs?
  
\item What general lower bounds can we give for our prediction models?

We feel that $\alpha_{GW} + O(\eps)$ is a natural barrier. One ``evidence'' we have is the following integrality gap for the SDP used in the partial information model; starting from a gap instance and an SDP solution exhibiting $\opt \leq \alpha_{GW} \cdot \sdp$ for the standard SDP (without incorporating revealed information), given labels for an $\eps n$ vertices, our new SDP simply fixes the positions of the corresponding $\eps n$ vectors, but doing that from the given SDP solution decreases the SDP value by at most $O(\eps)$ in expectation, which still yields $\opt \leq (\alpha_{GW} + O(\eps)) \sdp$. (Note that you can replace the SDP gap with any hypothetical gap instance for stronger relaxations.)

Given that the partial predictions model is easier than the noisy
predictions model and our entire algorithm for the partial model is
based on this SDP, this might be considered as a convincing lower
bound, but it would be nicer to have more general lower bounds against
all polynomial-time algorithms.

\item Our models only assume pairwise independence between vertices:
  can we extend our results to other ways of modeling correlations
  between the predictions? In addition to stochastic predictions, can we
  incorporate geometric predictions (e.g., in random graph models where
  the probability of edges depend on the proximity of the nodes)?
 
\end{enumerate}




\section*{Acknowledgments.} We thank Ola Svensson for enjoyable
discussions.

{\small
\bibliographystyle{alpha}
\bibliography{references}
}
\appendix
\section{Missing proofs for 2-CSPs}\label{sec:missing-proofs}

We obtain here the proof of \cref{thm:csps-high-degree}.
\paragraph*{Feasibility of the best assignment}\label{sec:csp-concentration}
As in \cref{lem:r-values-ok}, we first prove that, in expectation over the prediction $Y$, $x^*$ is a feasible solution to the program.

\begin{lemma}\label{lem:csp-concentration}
    Consider the settings of \cref{thm:csps-high-degree}.
    Then
    \begin{align*}
        \E \sum_{\substack{i\in[n]\\\Delta\textnormal{-narrow}}}&\left|\sum_{(w,c,\alpha)\in S_i}w P(c\circ x^{*\alpha \setminus i})\right|+\sum_{\substack{i\in[n]\\\Delta\textnormal{-wide}}}\left|\sum_{(w,c,\alpha)\in S_i\setminus \tilde{S}_i}w P(c\circ x^{*\alpha \setminus i})\right|\\
        +\sum_{\substack{i\in[n]\\\Delta\textnormal{-wide}}}&\left| 
        \sum_{(w,c,\alpha)\in \tilde{S}_i}w \Paren{P(c\circ x^{*\alpha \setminus i})-P(c\circ Z^{\alpha \setminus i})}\right|\leq  W(2\eta +O(1/\eps\sqrt{\Delta}))\,.
    \end{align*}
    \begin{proof}
        First,  by definition of $\Delta$-wide instance,$$\sum_{\substack{i\in[n]\\\Delta\textnormal{-narrow}}}\left|\sum_{(w,c,\alpha)\in S_i}w P(c\circ x^{*\alpha \setminus i})\right|\leq \eta W\,.$$
        Second, by definition for any $\Delta$-wide vertex $i$,
        \begin{align*}
            \Abs{\sum_{(w,c,\alpha)\in S_i\setminus\tilde{S}_i}wP(c\circ x^{*\alpha\setminus i})}\leq \eta W_i\,.
        \end{align*}
        Hence it remains to show
        \begin{align*}
            \E\sum_{\substack{i\in[n]\\\Delta\textnormal{-wide}}}&\left| 
        \sum_{(w,c,\alpha)\in \tilde{S}_i}w \Paren{P(c\circ x^{*\alpha \setminus i})-P(c\circ Z^{\alpha \setminus i})}\right|\leq  O(W/(\eps\sqrt{\Delta})\,.
        \end{align*}
        Now, recall that $\E[Y_i]=2\eps x^*_i$ and thus
         $\EE[Z] = x^*$. So for any $(c, \alpha)\in \cI\,,$
        $\E[P(c\circ Z^{\alpha})] = \E[P(c\circ x^{*\alpha})]$ by pair-wise independence of the predictions. 
        Thus it suffices to study, for each $\Delta$-wide $i$,
        $\Var\Paren{\sum_{(w,c,\alpha)\in S_i}w P(c\circ Z^{\alpha \setminus i})}\,.$
        To this end, notice that for any $\alpha,\alpha'\in S_i$ with $\alpha\cap\alpha'=\set{i}$ it holds
        \begin{align*}
            \E \Brac{P(c\circ Y^{\alpha\setminus i})P(c\circ Y^{\alpha'\setminus i})} =  \E \Brac{P(c\circ Y^{\alpha\setminus i})}\E \Brac{P(c\circ Y^{\alpha'\setminus i})}\,.
        \end{align*}
        Moreover, since $|\alpha|=2\,,$ there are at most $4$ distinct negation patterns. Therefore, by the AM-GM inequality
        \begin{align*}
            \Var\Paren{\sum_{(w,c,\alpha)\in \tilde{S}_i}w P(c\circ Z^{\alpha \setminus i})}
            &\leq \sum_{(w,c,\alpha)\in \tilde{S}_i} O(w^2)\Var\Paren{P(c\circ Z^{\alpha \setminus i})}\\
            &\leq \sum_{(w,c,\alpha)\in \tilde{S}_i} O\Paren{\frac{w^2}{\eps^2}}
        \end{align*}
        where we used the fact that entries of $Z$ are bounded by $1/\eps$ and the coefficients of a boolean predicate are bounded by $1$ (by Parseval's Theorem, see \cite{o2014analysis}).
        By construction of $\tilde{S}_i\,,$ each $(w,c,\alpha)\in \tilde{S}_i$ must satisfy $w\leq W_i/\Delta\,.$
        Using Holder's inequality
        \begin{align*}
            \Var\Paren{\sum_{(w,c,\alpha)\in \tilde{S}_i}w P(c\circ Z^{\alpha \setminus i})}\leq 
            O\Paren{\frac{W^2_i}{\Delta\cdot \eps^2}}\,.
        \end{align*}
        We can use this bound on the variance in combination with Chebishev's inequality to obtain, for $\lambda >0\,,$
        \begin{align*}
            \bbP \Paren{\Abs{
        \sum_{(w,c,\alpha)\in S_i}w \Paren{P(c\circ x^{*\alpha \setminus i})-P(c\circ Z^{\alpha \setminus i})}}\geq \lambda}
        &\leq O \Paren{\frac{W^2_i}{\eps^2\cdot \Delta\cdot\lambda^2}}\,.
        \end{align*}
        Let $\lambda := O(W_i/(\eps\sqrt{\Delta}))\,.$
        A peeling argument now completes the proof:
        \begin{align*}
            \E &\Brac{\Abs{
        \sum_{(w,c,\alpha)\in S_i}w \Paren{P(c\circ x^{*\alpha \setminus i})-P(c\circ Z^{\alpha \setminus i})}}}\\
        &\leq \lambda +\sum_{t\geq 0} 2^{t+1}\lambda \cdot \bbP  \Paren{\Abs{
        \sum_{(w,c,\alpha)\in S_i}w \Paren{P(c\circ x^{*\alpha \setminus i})-P(c\circ Z^{\alpha \setminus i})}}\geq 2^t \lambda}\leq O(\lambda)\,.
        \end{align*}
    \end{proof}
\end{lemma}

\paragraph*{Analysis of the algorithm}\label{sec:csp-analysis-algorithm}
We can use \cref{lem:csp-concentration} to obtain our main theorem for CSPs.

\begin{proof}[Proof of \cref{thm:csps-high-degree}]
    We follow closely the proof of \cref{lem:algo-value}.
    Consider one of the assignments $X\in \set{\pm 1}^n$ found in the second step of the algorithm. Recall $\hat{x}\in[-1,+1]^n$ denotes the optimal fractional solution found by the algorithm.  We may rewrite for each $\Delta$-wide $i$ and $(w,c,\alpha)\in \tilde{S}_i$
    \begin{align}\label{eq:csp-expectation}
         \sum_{\substack{i\in [n]\\ \Delta\textnormal{-wide}}}\sum_{(w,c,\alpha)\in \tilde{S}_i}wP(c\circ X^{\alpha}) =  & \sum_{\substack{i\in [n]\\ \Delta\textnormal{-wide}}}\sum_{(w,c,\alpha)\in \tilde{S}_i}w\left[ P(c\circ (\hat{x}_i\cdot Z^{\alpha\setminus i})) \right.\nonumber\\
          &+\left.P(c\circ X^{\alpha})- P(c\circ \hat{x}^{\alpha})\right.\nonumber\\
         &+ \left.P(c\circ \hat{x}^{\alpha}) -P(c\circ (\hat{x}_i\cdot Z^{\alpha\setminus i}))\right]\,.
    \end{align}
    We bound  each term in \cref{eq:csp-expectation} separately.
    First, notice that by Markov's inequality and \cref{lem:csp-concentration}, with probability $0.99\,,$ $x^*$ is a feasible solution to the LP. Conditioning on this event $\cE$
    \begin{align*}
        \sum_{\substack{i\in [n]\\ \Delta\textnormal{-wide}}}\sum_{(w,c,\alpha)\in \tilde{S}_i}w P(c\circ (\hat{x}_i\cdot Z^{\alpha\setminus i}))\geq& \sum_{\substack{i\in [n]\\ \Delta\textnormal{-wide}}}\sum_{(w,c,\alpha)\in \tilde{S}_i}w P(c\circ (x^*_i\cdot Z^{\alpha\setminus i}))\\
        = & \sum_{\substack{i\in [n]\\ \Delta\textnormal{-wide}}}\sum_{(w,c,\alpha)\in \tilde{S}_i}w P(c\circ x^{*\alpha})\\
        +& \sum_{\substack{i\in [n]\\ \Delta\textnormal{-wide}}}\sum_{(w,c,\alpha)\in \tilde{S}_i}w \Paren{P(c\circ (x^*_i\cdot Z^{\alpha\setminus i})) - P(c\circ x^{*\alpha}) }
    \end{align*}
    By Holder's inequality and the fact that $x^*$ is feasible, for $\Delta$-wide $i$,
    \begin{align*}
        \sum_{\substack{i\in [n]\\ \Delta\textnormal{-wide}}}\sum_{(w,c,\alpha)\in \tilde{S}_i}w& \Paren{P(c\circ (x^*_i\cdot Z^{\alpha\setminus i})) - P(c\circ x^{*\alpha}) }\\
        \leq \sum_{\substack{i\in [n]\\ \Delta\textnormal{-wide}}}&\Abs{\sum_{(w,c,\alpha)\in \tilde{S}_i}w \Paren{P(c\circ Z^{\alpha\setminus i}) - P(c\circ x^{*\alpha\setminus i}) }}
        \leq (O(\eps')+2\eta)W\,.
    \end{align*}
    Since by construction $\hat{x}$ is feasible, another application of Holder's inequality also yields the following bound on the third term,
    \begin{align*}
        \sum_{\substack{i\in [n]\\ \Delta\textnormal{-wide}}}\sum_{(w,c,\alpha)\in \tilde{S}_i}w \Paren{P(c\circ (\hat{x}_i\cdot Z^{\alpha\setminus i})) - P(c\circ \hat{x}^{\alpha}) }
        \leq (O(\eps')+2\eta)W\,.
    \end{align*}
    For the second term in \cref{eq:csp-expectation}, by construction of $X$ we have $\E \Brac{P(c\circ X^{\alpha})\,|\,\cE}= P(c\circ \hat{x}^{\alpha})$.
    Combining the three bounds, we get that  
    \begin{align*}
        \opt_\cI\geq E \Brac{\left.\frac{1}{W}\sum_{\substack{i\in [n]\\ \Delta\textnormal{-wide}}}\sum_{(w,c,\alpha)\in \tilde{S}_i}w  P(c\circ X^{\alpha}) \right|\cE }\geq \opt_\cI - (O(\eps')+4\eta)\,.
    \end{align*}
    Applying Markov's inequality on the random variable $\opt_\cI - \frac{1}{W}\sum_{\substack{i\in [n]\\ \Delta\textnormal{-wide}}}\sum_{(w,c,\alpha)\in \tilde{S}_i}w  P(c\circ X^{\alpha})$, we get
    \begin{align*}
        \bbP & \Paren{\left.\frac{1}{W}\sum_{\substack{i\in [n]\\ \Delta\textnormal{-wide}}}\sum_{(w,c,\alpha)\in \tilde{S}_i}w  P(c\circ X^{\alpha}) \leq \opt_\cI - (O(\eps')+5\eta) \right|\cE }\leq \frac{1}{1+\eta}
    \end{align*}
    The theorem follows since we sample $O(1/\eta)$ independent assignments $X$ and pick the best.
\end{proof}

\end{document}